%% file: main.tex
\documentclass[twoside,reqno,twocolumn]{article}
\usepackage{ltexpprt}

\usepackage{amsmath}
\usepackage{amssymb}
\usepackage{algorithm}
\usepackage{algpseudocode}
\usepackage{varwidth}
\usepackage{braket}
\usepackage{graphicx}
\usepackage{subcaption}
\usepackage{pgfplots}
\usepackage{cleveref}
\allowdisplaybreaks

\pgfplotsset{compat = newest}

\Crefname{ALC@unique}{Line}{Lines} 

\begin{document}

\title{Modified Iterative Quantum Amplitude Estimation is Asymptotically Optimal}

\author{Shion Fukuzawa
\thanks{University of California, Irvine {\tt fukuzaws@uci.edu, christh9@uci.edu, irani@ics.uci.edu, jzion@uci.edu}}
\and Christopher Ho$^*$
\and Sandy Irani$^*$
\and Jasen Zion$^*$
}

\date{}

\maketitle


\fancyfoot[R]{\scriptsize{Copyright \textcopyright\ 2023 by SIAM\\
Unauthorized reproduction of this article is prohibited}}

\begin{abstract} \small\baselineskip=9pt
In this work, we provide the first QFT-free algorithm for Quantum Amplitude Estimation (QAE) that is asymptotically optimal while maintaining the leading numerical performance. QAE algorithms appear as a subroutine in many applications for quantum computers. The optimal query complexity achievable by a quantum algorithm for QAE is $O\left(\frac{1}{\epsilon}\log \frac{1}{\alpha}\right)$ queries, providing a speedup of a factor of $1/\epsilon$ over any other classical algorithm for the same problem. The original algorithm for QAE utilizes the quantum Fourier transform (QFT) which is expected to be a challenge for near-term quantum hardware. To solve this problem, there has been interest in designing a QAE algorithm that avoids using QFT. Recently, the iterative QAE algorithm (IQAE) \cite{grinko2021a} was introduced with a near-optimal $O\left(\frac{1}{\epsilon}\log \left(\frac{1}{\alpha} \log \frac{1}{\epsilon}\right)\right)$ query complexity and small constant factors. In this work, we combine ideas from the preceding line of work to introduce a QFT-free QAE algorithm that maintains the asymptotically optimal $O\left(\frac{1}{\epsilon}\log \frac{1}{\alpha}\right)$ query complexity while retaining small constant factors. We supplement our analysis with numerical experiments comparing our performance with IQAE where we find that our modifications retain the high performance, and in some cases even improve the numerical results.
\end{abstract}

\input{./introduction}

\input{./algorithm}

\input{./theorem}

\input{./experiments}

\input{./conclusion}

\bibliographystyle{unsrt}
\bibliography{references}

\newpage

\input{./appendix}

\end{document}

%% file: introduction.tex
\section{Introduction}

As we see the rise of the first commercially available programmable quantum computers, quantum algorithms designed over the last few decades that promise various speedups over their classical counterparts are gaining prevalence. However, these devices are still not very large in scale, and use noisy qubits which make it challenging to implement useful algorithms that would require a large number of robust qubits in highly complex circuits using many controlled gates. Because of these restrictions, algorithms that can be run effectively in this so-called NISQ era \cite{preskill2018} have been of great interest, where the objective is to find algorithms which can either reduce the number of qubits, the depth of the circuit, and/or the number of controlled gates required.

Two decades ago, Brassard, Hoyer, Mosca, and Tapp \cite{brassard2002} introduced the problem of quantum amplitude estimation. In this problem, we are given a unitary operator $\mathcal{A}$ that acts on $n + 1$ qubits such that 

\begin{equation}\label{eq:def-A}
    \mathcal{A}\ket{0}_n \ket{0} = \sqrt{1 - a}\ket{\psi_0}_n \ket{0} + \sqrt{a}\ket{\psi_1}_n \ket{1},
\end{equation}

where $a \in [0,1]$ is unknown and $\ket{\psi_0}$, $\ket{\psi_1}$ are two normalized states. Given this setting, the objective is to find a good approximation for $a$ with as few applications of $\mathcal{A}$ as possible.
They presented the problem as a generalization of Grover's algorithm \cite{grover1996}, and demonstrate that techniques used by Grover could be used to solve QAE. Brassard et al.'s algorithm enables a quadratic speedup for many applications that are solved classically using Monte Carlo simulation, which itself proves useful in finance for applications such as risk analysis \cite{woerner2019, egger2019} or option pricing \cite{rebentrost2018, stamatopoulos2020, zoufal2019}, as well as other general tasks like numerical integration \cite{montanaro2015}. Furthermore, amplitude estimation is a widely used subroutine in many quantum algorithms, making it an algorithm of high interest to implement in the near future.

Though established as a fundamental algorithm in the field, a challenge in implementing this algorithm is that it requires the ability for the quantum device to perform a large number of controlled gates as it requires the ability to implement controlled unitary gates as well as the quantum Fourier transform (QFT), both of which have been challenging to realize on physical devices. Due to this, there has recently been an increased interest in algorithms that can match the performance of QAE without requiring QFT \cite{grinko2021a, aaronson2020, suzuki2020, wie2019, rall2022, zhao2022a, plekhanov2022}. It is unlikely that this is sufficient for QAE to become a NISQ compatible algorithm, but our algorithm significantly reduces the width required for the circuit. This implies that if physical devices that can run deep circuits, these algorithms can be beneficial even with a small number of qubits. 

Aaronson and Rall \cite{aaronson2020} introduced QAES, the first QFT-free version of QAE which matches the query complexity of the original algorithm. Though asymptotically tight, the constant overhead of QAES will likely require better quantum hardware to be implemented effectively. Within the same year, two other QFT-free algorithms were introduced by Suzuki et al. \cite{suzuki2020} and Wie \cite{wie2019}. Suzuki et al.'s algorithm, which is often referred to in the literature as MLAE (Maximum Likelihood Amplitude Estimation), replaces QFT with a classical maximum likelihood estimation to analyze the combined results of the quantum circuit. Numerical experiments were performed to demonstrate the performance of the algorithm, but it is not clear how to prove an upper bound on query complexity for their algorithm. 

In 2021, Grinko, Gacon, Zoufal, and Woerner \cite{grinko2021a} introduce Iterative Quantum Amplitude Estimation (IQAE) which helped remedy the weaknesses of the three above algorithms. Its experimental performance beats that of MLAE, and the query complexity was proven to be just within a log-log factor of the desired complexity. As the name suggests, IQAE is an iterative algorithm, where the approximation of the target value is improved at each round. 

In their work, it was shown that IQAE is the best performer numerically amongst the QFT-free algorithms mentioned above \cite{grinko2021a}. In this work, we show that a modified version of IQAE matches the lower bound on the query complexity of amplitude estimation, while also maintaining the competitive numerical performance that IQAE provides. 
To achieve this, instead of requiring the same probability of success in each round, we adopt a more intricate sequence for the target probability of success in each round to compensate for the fact that the trials in later rounds require more queries. This allows us to shave off the extra log-log factor from the bound of the original IQAE, and match the lower bound query complexity of $O\left(\frac{1}{\epsilon}\log \frac{1}{\alpha}\right)$ for amplitude estimation without requiring QFT. Here, the parameter $\epsilon$ is the desired precision of the estimate, and $1 - \alpha$ is the probability that the true amplitude $a$ is $\epsilon$ away from the output estimate. In addition to matching the theoretical lower bound, we also demonstrate that our algorithm achieves comparable and, in certain important cases, better numerical performance than the original IQAE. 

A recent development of interest to this paper is Rall and Fuller's work \cite{rall2022} which showed that by a modification in the overshooting condition laid out in IQAE, it was possible to get an increase in query complexity by a factor of two. This amounts to a choice in a hyperparameter for the algorithm. We discuss our experimental observations at the end of section 2.1 and in the results section, but we find that a simple selection for hyperparameter shows a similar improvement in numerical query complexity across all variants of IQAE. There may be other ways to analyze the optimal hyperparameter depending on the experimental setup and other factors, but we believe our suggestion is a reasonably good choice. Furthermore, our work resolves a question they open demonstrating that it is possible to achieve optimal query complexity without a large numerical penalty. 

Other notable works on amplitude estimation include Zhao et al. \cite{zhao2022a} who showed that they are able to reduce the classical cost of the algorithm.

\subsection{Quantum Amplitude Estimation. }

A core procedure for all amplitude estimation algorithms mentioned in this paper is amplitude amplification. To describe this, we begin by defining $\theta_a \in [0, \pi/2]$ as the angle that satisfies $\sin^2(\theta_a) = a$, where $a$ is defined in (\ref{eq:def-A}). Then, given any unitary $\mathcal{A}$ defined by (\ref{eq:def-A}) we see that 
\begin{equation}
    \mathbb{P}[\ket{1}] = \sin^2(\theta_a). 
\end{equation}
As described by Brassard et al. \cite{brassard2002} and Grover \cite{grover1996}, we can define and easily construct the operator $Q := \mathcal{A}\mathcal{S}_0\mathcal{A}^\dagger \mathcal{S}_{\psi_0}$, where $\mathcal{S}_0 := \mathbb{I} - 2 \ket{\psi_0}\bra{\psi_0} \otimes \ket{0}\bra{0}$ and $\mathcal{S}_{\psi_0} := \mathbb{I} - 2 \ket{0}_{n+1}\bra{0}_{n+1}$. It has been shown that if we measure the circuit $Q^k\mathcal{A}\ket{0}_n \ket{0}$, we have
\begin{equation}
    \mathbb{P}[\left|1\right>] = \sin^2 ((2k+1)\theta_a). 
\end{equation}
This process of boosting the probability of measuring a $\ket{1}$ state lies at the core of all QAE algorithms mentioned in this paper, leading to the various speed-ups over naive sampling of the distribution defined by $\mathcal{A}$. 

To compare our results to the original IQAE paper, we define the query complexity of a QAE algorithm as the total number of times this operator $Q$ is applied. Note that this definition of query complexity can be easily modified to count the number of times $\mathcal{A}$ is applied.

%% file: algorithm.tex
\section{Algorithm}

\subsection{Overview of IQAE. }

We will start with an overview of the IQAE algorithm and then describe our modifications to improve the asymptotic query complexity. Note that our description of IQAE is somewhat different from \cite{grinko2021a}. We believe that the explanation
given here distills the essential features of the algorithm and clarifies how IQAE differs from previous approaches.

Like the other QFT-free  algorithms, the IQAE algorithm uses the quantum computer to approximate $\mathbb{P}[\ket{1}] = \sin^2((2k+1)\theta_a)$ for the last qubit in $Q^k\mathcal{A}\ket{0}_n\ket{0}$ for different powers of $k$. Thus our measurements of the circuits provide us with an approximation of $\theta_a$, which can be converted to an approximation of $a$ within a desired precision. The algorithm we describe closely follows the structure of IQAE, but we will state all elements of the algorithm and proofs for completeness.

The algorithm has two nested loops and we refer to each iteration of the outer loop as a {\em round}.
As with all the QFT-free QAE algorithms, IQAE maintains a confidence interval $[\theta_l, \theta_u]$ for $\theta_a$
so that it can be proven that $\theta_a \in [\theta_l, \theta_u]$ with a certain probability.
QAES explicitly narrows down the size of the interval by a constant factor in each iteration.
IQAE makes progress towards narrowing the interval more indirectly.
In a given round $i$, the algorithm maintains an odd integer $K_i$ and implements measurements of the
form $Q^{k_i}\mathcal{A}\ket{0}_n\ket{0}$, for $K_i = 2 k_i + 1$.
IQAE  maintains  the invariant that the scaled interval $[K_i\theta_l, K_i\theta_u]$ lies in a single quadrant, described mathematically as 
\begin{equation}
\label{eq:invariant}
    \left\lfloor \frac{K_i \theta_l}{\pi/2} \right\rfloor = \left\lceil \frac{K_i \theta_u}{\pi/2} \right\rceil - 1.
\end{equation}
We refer to the value $R_i = \lfloor K_i \theta_l / (\pi/2) \rfloor$ as the {\em quadrant count}, which is
the number of complete quadrants passed through starting at an angle of $0$ and slowly increasing the angle until $K_i \theta_l$ or $K_i \theta_u$ is reached. The rounding is done slightly differently for the upper bound
than the lower bound. The only difference occurs when $ K_i \theta_l / (\pi/2) $ or $ K_i \theta_u / (\pi/2) $
are integers. In this case, we would like the upper bound to map to the previous quadrant.
This ensures, for example, that for $[\theta_l, \theta_u] = [\pi, 5 \pi/4]$ and
$[\theta_l, \theta_u] = [5 \pi/4, 3 \pi/2]$ both angles will have a quadrant count of $2$ as can be seen in Figure \ref{fig:quadrant-disambiguation}.

\begin{figure}[!ht]
    \centering
    \includegraphics[width=.8\linewidth]{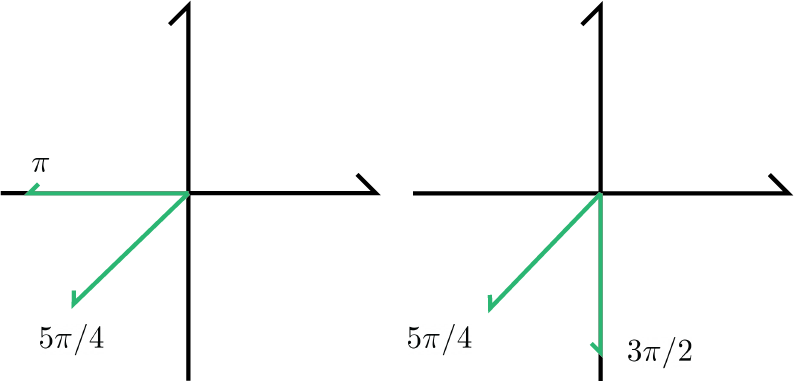}\\
    \caption{\em Edge cases for intervals of $[\theta_l, \theta_u]$. For each of the cases pictured above, the quadrant counts $R_i$ will be 2. Note that the quadrant count begins at 0. Our rounding scheme ensures that this is indeed the case, and the code we used further handles numerical floating errors that were occurring in these regions. }
    \label{fig:quadrant-disambiguation}
\end{figure}

Note that Invariant (\ref{eq:invariant}) implies that 
$K_i \theta_l$ and $K_i \theta_u$ can both be expressed as
\begin{align*}
    K_i \theta_l & = R_i \cdot (\pi/2) + \gamma_l& \mbox{for}~\gamma_l \in [0, \pi/2) \\
    K_i \theta_u & = R_i \cdot (\pi/2) + \gamma_u& \mbox{for}~\gamma_u \in (0, \pi/2] 
\end{align*}
The invariant also implies that $\theta_u - \theta_l \le \pi/2K_i$.
The algorithm makes progress from round to round by increasing the value of $K_i$ for which the invariant holds. The analysis shows that
the value of $K_i$ will increase by at least factor of $3$ in each round.

If Invariant (\ref{eq:invariant}) holds and
$\theta_a \in [\theta_l, \theta_u]$, then $K_i \theta_a$ can also be expressed as 
$R_i \cdot (\pi/2) + \gamma_a$, for $\gamma_a \in [0, \pi/2]$.
Since $K_i$ is an odd integer, we can measure $Q^{k_i}\mathcal{A}\ket{0}_n\ket{0}$, for $k_i = (K_i -1)/2$, which will return a state ending in $\ket{1}$ with probability $\sin^2 (K_i \theta_a)$.
The value $\sin^2 (K_i \theta_a)$ is equal to $\sin^2 (\gamma_a)$ or $1 - \sin^2 (\gamma_a)$, depending on whether the quadrant count $R_i$ is odd or even, as can be seen in Figure \ref{fig:gamma-alpha-conversion}. 
Using repeated measurements, we obtain a confidence interval $[a^{min}, a^{max}]$
for $\sin^2 (K_i \theta_a)$ which can be in turn translated
to a confidence interval $[\gamma^{min}, \gamma^{max}]$ for $\gamma_a$.
The Chernoff-Hoeffding inequality is used to upper bound the probability
that $\sin^2 (K_i \theta_a) \not\in [a^{min}, a^{max}]$, which is equal to the probability that $\gamma_a \not\in [\gamma^{min}, \gamma^{max}]$.
Table \ref{table:theta-i-ci} shows  the mapping from the confidence interval $[a^{min}, a^{max}]$
 to the confidence interval $[\gamma^{min}, \gamma^{max}]$
based on the parity of the quadrant count $R_i$.

\begin{figure}[!b]
    \centering
    \includegraphics[width=.8\linewidth]{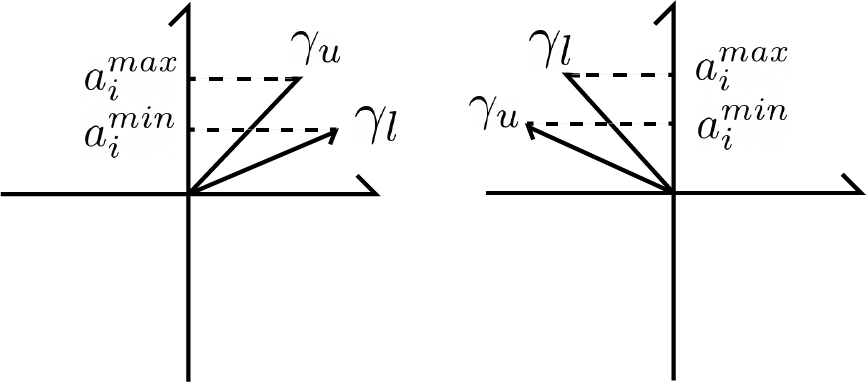}\\
    \caption{\em If the interval $[\gamma_l, \gamma_u]$ lies in an even quadrant, as $\gamma$ increases, the value $\sin^2(\gamma_i)$ increases for $i \in \{l, u\}$. On the other hand, for odd quadrants, as $\gamma$ increases, the value $\sin^2(\gamma_i)$ decreases. This drove the calculations in the conversion described in Table \ref{table:theta-i-ci}. }
    \label{fig:gamma-alpha-conversion}
\end{figure}

The new interval for $\gamma_a$ gives rise to new values  for $\theta_l$ and $\theta_u$:
$$\theta_l \leftarrow \frac{R_i (\pi/2) + \gamma_{min}}{K_i}~~~~~~\theta_u \leftarrow \frac{R_i (\pi/2) + \gamma_{max}}{K_i}$$
Invariant (\ref{eq:invariant}) trivially holds for the new $\theta_l$ and $\theta_u$
as well as the old value of $K_i$. However, if the interval $[\theta_l, \theta_u]$ has
narrowed sufficiently, then the invariant may hold for a larger value of $K$.
Note that if we have a fixed target for the probability that
$\gamma_a \in [\gamma^{min}, \gamma^{max}]$, the Chernoff-Hoeffding bound can be used to trade off the number of queries
made against the size of the interval.
For round $i$, we denote by $\alpha_i$ the desired upper bound on
probability of failure conditioned on success in all of the 
previous rounds. 
The heart of the analysis is to show
that the expression $N_{max}$, given in Line 11 of Algorithm 1,
is sufficient to guarantee that Invariant (\ref{eq:invariant}) holds with a new
value of $K$ satisfying $K \geq 3 K_i$.

\begin{table}[!t]
\centering
\begin{tabular}{|c|c|c|}
\hline
Quad. $R$& $\gamma_i^{min}$                   & $\gamma_i^{max}$                   \\ \hline
Even        & $\sin^{-1} \sqrt{a_i^{min}}$         & $\sin^{-1} \sqrt{a_i^{max}}$         \\ \hline
Odd       & $-\sin^{-1} \sqrt{a_i^{max}} + \frac{\pi}{2}$  & $-\sin^{-1} \sqrt{a_i^{min}} + \frac{\pi}{2}$  \\ \hline
\end{tabular}
\caption{\em A table describing the conversion from $a_i^{min}$, $a_i^{max}$ to $\gamma_i^{min}$, $\gamma_i^{max}$ respectively. $\sin^{-1}$ denotes the arcsin function. The top right and bottom left quadrants correspond to an even quadrant count. The top left and bottom right quadrants correspond to an odd quadrant count.}
\label{table:theta-i-ci}
\end{table}

\begin{algorithm}
\caption{Modified IQAE}
\label{alg:modified-iqae}
\begin{algorithmic}[1]
\Require $\epsilon > 0$, $\alpha > 0$, $N_{\text{shots}} \in \mathbb{N}$, Unitary $\mathcal{A}$
\State $i = 0$, $k_i = 0$ 
\State $[\theta_l, \theta_u] = [0, \pi/2]$ 
\State $K_{max} := \frac{\pi}{4\epsilon}$ 
\While{$\theta_u - \theta_l > 2\epsilon$}
\State $N := 0$, $i := i + 1$, $k_i = k_{i-1}$
\State $K_i := 2k_i + 1$
\State $\alpha_i := \frac{2\alpha}{3} \frac{K_i}{K_{max}}$
\State $N_i^{max} = \frac{2}{\sin^2(\pi/21)\sin^2(8\pi/21)}\ln\left(\frac{2}{\alpha_i}\right)$
\State \begin{varwidth}[t]{.95\linewidth} 
$R_i = \lfloor \frac{K_i\theta_l}{\pi/2} \rfloor$ \algorithmiccomment{The number of quadrants passed to get to current angle}
\end{varwidth}
\While{$k_i = k_{i-1}$}
\State Prepare circuit $Q^{k_i}\mathcal{A}\ket{0}_n\ket{0}$
\State Measure $\min\{N_{\text{shots}}, N_{i}^{max}-N\}$ times 
\State \begin{varwidth}[t]{.95\linewidth} Combine results in this round to approximate $\hat{a}_i \approx \mathbb{P}\left[\ket{1}\right]$. \end{varwidth}
\State $N = N + \min\{N_{\text{shots}}, N_{i}^{max}-N\}$
\State $\epsilon_{a_i} := \sqrt{\frac{1}{2N} \ln\left(\frac{2}{\alpha_i}\right)}$ 
\State $a_i^{\max} = \min(1, \hat{a}_i + \epsilon_{a_i})$
\State $a_i^{\min} = \max(0, \hat{a}_i - \epsilon_{a_i})$
\State \begin{varwidth}[t]{.88\linewidth} Calculate the confidence interval $[\gamma_i^{\min}, \gamma_i^{\max}]$ for $\{K_i\theta_a\}$ from $[a_i^{\min}, a_i^{\max}]$ and $R_i$ according to table \ref{table:theta-i-ci} \end{varwidth}
\State $\theta_l = \frac{R_i \cdot \pi/2 + \gamma_i^{\min}}{K_i}$
\State $\theta_u = \frac{R_i \cdot \pi/2 + \gamma_i^{\max}}{K_i}$
\If{$\theta_u-\theta_l < 2\epsilon$}
\State Break
\EndIf
\State $k_{i} := \texttt{FindNextK}(k_{i}, \theta_l, \theta_u)$
\EndWhile
\EndWhile
\State $[a_l, a_u] = [\sin^2(\theta_l), \sin^2(\theta_u)]$
\State \Return $[a_l, a_u]$
\end{algorithmic}
\end{algorithm}

\begin{algorithm}
\caption{FindNextK}
\label{alg:find-next-k}
\begin{algorithmic}[1]
\Require $k_{i}$, $\theta_l$, $\theta_u$
\State $K_{i} = 2k_{i} + 1$
\State $K = \left\lfloor \frac{\pi/2}{\theta_u - \theta_l} \right\rfloor$ 
\If {$K$ is even}
\State $K = K - 1$
\EndIf
\While{$K \geq 3K_{i}$}
\If {$\left\lfloor \frac{K\theta_l}{\pi/2} \right\rfloor = \left\lceil \frac{K\theta_u}{\pi/2} \right\rceil - 1$}
\State \Return $(K - 1)/2$
\EndIf
\State $K = K - 2$
\EndWhile
\State \Return $k_{i}$
\end{algorithmic}
\end{algorithm}

Note that if a round begins by obtaining $N_{max}$ measurements, the inner loop is unnecessary.
The algorithm will advance to the next round with the same upper bound on the probability
of failure, and the asymptotic analysis for the number of queries still holds.
The numerical performance of the algorithm can be improved by starting with a smaller number of
samples and progressively obtaining more samples until the interval $[\gamma^{min}, \gamma^{max}]$
has narrowed enough that  the
algorithm succeeds in advancing to the
next round. $N_{shots}$ is the number of additional measurements taken in 
each iteration of the inner loop.
The procedure FindNextK, takes the current values for $\theta_l$ and $\theta_u$
and finds the smallest odd integer $K$ in the range from $3K_i$ to $(\pi/2)/(\theta_u - \theta_l)$
such that $K \theta_l$ and $K \theta_u$ have the same quadrant count $R$. If no such $K$ exists,
the old value $K_i$ is returned.
Lemmas \ref{lemma:increase-angle} and \ref{lemma:round-termination} show that 
FindNextK will succeed in finding a new value for $K$
by the time the total number of measurements taken within a round reaches
$N_{max}$. The inner loop of the algorithm is illustrated in Figure \ref{fig:algorithm1}.

Note that in a given round, the old values for $\theta_l$ and $\theta_u$ are only used to find the
common quadrant count $R$. Once $R$ is determined, new values for $\theta_l$ and $\theta_u$ are derived
from the confidence interval given by the Chernoff-Hoeffding bound, which is a function of the target
probability of success and the samples taken in the round. 

IQAE also implements the Clopper-Pearson method \cite{clopper1934} for calculating a confidence interval. This method gives a more accurate confidence interval with fewer samples relative to the Chernoff-Hoeffding bound but is not amenable to a rigorous analysis. Our proof can be adapted in a way analogous to the way Grinko et al. \cite{grinko2021a} have in the appendix of their work. Our results surrounding this method are demonstrated in the numerical results section through experiment. 

We use the variable $N_\text{shots}$ as a way to gradually reach $N_{i}^{max}$ for better numerical performance. In each iteration during a round, the algorithm takes $N_\text{shots}$ samples from the quantum circuit to compute the confidence interval. If the algorithm achieves the necessary sampling conditions for the round without needing all $N_{i}^{max}$ shots, we would prefer it to proceed. In the original IQAE algorithm \cite{grinko2021a}, they defined an overshooting condition to transition between an ``early'' vs. ``late'' stage of the algorithm. These stages are a rough heuristic characterization of the circuit depth (size of $k_i$) to be sampled. When the circuit is deep, the query cost becomes very large which can blow up the numerical query complexity. Instead of taking $N_\text{shots}$ samples per iteration, they introduce a heuristic telling the algorithm to take a smaller number of shots at a time once the circuit reaches a certain depth. Rall and Fuller \cite{rall2022} refine this heuristic in what they call the ChebAE algorithm, simplifying the reasoning behind when this change should occur as demonstrated by their improved numerical performance. Since we distinguish the early and late stages of the algorithm gradually by varying the required confidence interval size, we show in the results section that choosing $N_\text{shots} = 1$ is able to demonstrate a similar numerical improvement over our initial experiments using $N_\text{shots} = 100$. 

\begin{figure*}[!ht]
    \centering
    \includegraphics[width=0.9\linewidth]{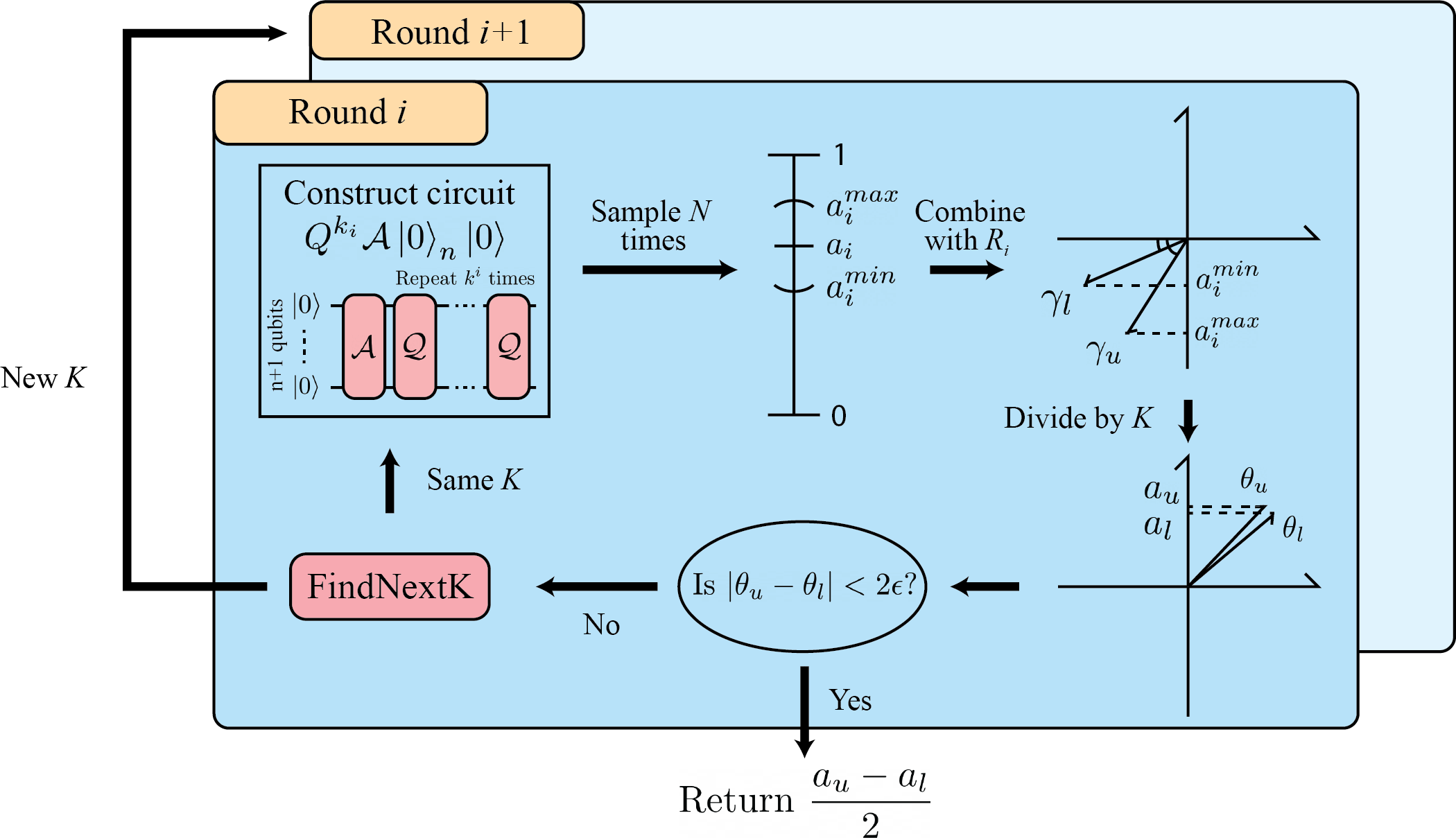}
    \caption{\em An illustration of the inner loop of Algorithm \ref{alg:modified-iqae}. In a given round $i$, we store a variable $k_i$ and $K_i := 2k_i + 1$. We first construct the circuit $Q^{k_i}\mathcal{A}\left|0\right>_n \left|0\right>$ and sample from it $N_{shots}$ times, which gives us a confidence interval for the variable $a_i$. The endpoints of this confidence interval are then converted to variables $\gamma_i^{min}$ and $\gamma_i^{max}$. Dividing these values by $K_i$ gives us a confidence interval for $a$, and we exit the algorithm once the desired precision is reached. If not, we call FindNextK (Algorithm \ref{alg:find-next-k}) and continue sampling from the circuit if it returns the same K, or advance to round $i + 1$ if it returns a new K.}
    \label{fig:algorithm1}
\end{figure*}

\subsection{Modifications to IQAE. }

The probability that the algorithm fails in round $i$ conditioned on success in all the previous rounds
is upper bounded by $\alpha_i$. Using the union bound, it can be shown that the probability of failure overall is at most
$\sum_{i=1}^t \alpha_i$, where $t$ is the total number of rounds the algorithm goes through before returning.
The original IQAE algorithm derives an upper bound $T$ on the total number of rounds and sets the target 
probability of failure to be the same for all rounds: $\alpha_i = \alpha/T$.
However, as the circuit depth increases exponentially between rounds, using a fixed probability of failure forces the query complexity to depend heavily on the depth of the final round. 
A measurement in round $i$ is applied to $Q^{k_i}\mathcal{A}\ket{0}_n\ket{0}$, which requires $K_i = 2k_i+1$ applications of $\mathcal{A}$. Modified IQAE avoids this issue by increasing the probability of failure in each round
while making sure that $\sum_{i=1}^t \alpha_i$ is still upper bounded by $\alpha$. This has the effect of ensuring that fewer measurements are required in the later rounds.
We define $\alpha_i$ as
$\alpha_i = \frac{2\alpha}{3} \frac{K_i}{K_{max}}$, where $K_{max} = \pi/4 \epsilon$
is an upper bound for the value of $K_t$, the final value in the sequence of $K_i$'s, derived such that the algorithm is guaranteed to terminate. Effectively, we are scaling $\alpha_i$ relative to the depth of the circuit that will be required in that round. 
It turns out that this new schedule for the $\alpha_i$'s allows us to shave off a factor of $\log \log \epsilon$ from the
upper bound on the number of queries.
A similar idea was used in QAES where as the circuit depth increased, the algorithm gradually loosened the required width for the confidence interval \cite{aaronson2020}.

%% file: theorem.tex
\section{Main Theorem}

In this section, we formally state our theorem and proof. 

\begin{theorem}\label{thm:algorithm}
Given a confidence level $1 - \alpha \in (0, 1)$, a target accuracy $\epsilon > 0$, and an $(n + 1)$-qubit unitary $\mathcal{A}$ satisfying 

\[\mathcal{A}\ket{0}_n\ket{0} = \sqrt{a}\ket{\psi_0}\ket{0} + \sqrt{1-a}\ket{\psi_1}\ket{1}\]

where $\ket{\psi_0}$ and $\ket{\psi_1}$ are arbitrary $n$-qubit states and $a \in [0, 1]$, modified IQAE outputs a confidence interval for $a$ that satisfies

\[\mathbb{P}[a \not \in [a_l, a_u]] \leq \alpha\]

where $a_u - a_l < 2\epsilon$, leading to an estimate $\hat{a}$ for $a$ such that $|a - \hat{a}| < \epsilon$ with a confidence of $1 - \alpha$, using $O\left(\frac{1}{\epsilon}\log \frac{1}{\alpha}\right)$ applications of $\mathcal{A}$. 
\end{theorem}

\subsection{Upper Bound on the Number of Rounds. }

The algorithm begins with $K_0 = 1$. By definition, the algorithm advances to the next round
when FindNextK successfully finds a new value of $K$, and FindNextK only returns a new
value if the new $K$ is at least $3 K_i$, where $K_i$ is the current value of $K$.
Therefore, $K_{i} \ge 3 K_{i-1}$ for each round.
The upper bound proven below for the maximum value of $K$ establishes that the number of rounds
will be at most $\log_3 (\pi/ 4 \epsilon)$. 

\begin{lemma}\label{lemma:K-i-upperbound}
$K_i \leq \frac{\pi}{4\epsilon}$ for all $i$. 
\end{lemma}

\begin{@proof}
By the termination conditions of our algorithm, if FindNextK is called at the end of a round, then we know $\theta_u - \theta_l \geq 2\epsilon$. Furthermore, by the invariant condition for the scaled confidence interval described in equation (\ref{eq:invariant}), if this call to FindNextK returns a new $K_i$, it must be the case that  
\begin{equation}
    K_i \leq \frac{\pi/2}{\theta_u - \theta_l}.
\end{equation}
We show in lemma \ref{lemma:increase-angle} that such a $K_i$ always exists. Combining the above, we get 
\begin{equation}
    K_i \leq \frac{\pi/2}{2\epsilon} = \frac{\pi}{4\epsilon}.
\end{equation}
\end{@proof}

\subsection{Proof that Each Round Successfully Terminates. }

Lemma \ref{lemma:round-termination} establishes that each round will successfully terminate by proving
that by the time $N$ reaches $N_{max}$, FindNextK will successfully find a new value for $K$
which allows the algorithm to advance to the next round. 
The proof of Lemma  \ref{lemma:round-termination} depends on
the following lemma whose result can be stated independent of the algorithm. The proof 
of Lemma \ref{lemma:increase-angle} is slightly involved and distracting from understanding the dynamics of the algorithm, so it has been moved to the appendix.

\begin{lemma}\label{lemma:increase-angle}
Given $\theta_a, \theta_b$ such that 
\begin{enumerate}
    \item $\theta_b - \theta_a > 0$,
    \item $|\sin^2 \theta_b - \sin^2 \theta_a| \leq \sin \frac{\pi}{21} \sin \frac{8\pi}{21}$, 
    \item the two angles lie in the same quadrant, i.e., 
    \[\left\lfloor \frac{\theta_a}{\pi/2} \right\rfloor = \left\lceil \frac{\theta_b}{\pi/2}\right\rceil - 1,\]
\end{enumerate}
then there exists an odd integer $q \in \left[3, \frac{\pi/2}{\theta_b - \theta_a}\right]$ such that 
\begin{equation}
    \left\lfloor \frac{q\theta_a}{\pi/2} \right\rfloor = \left\lceil \frac{q\theta_b}{\pi/2} \right\rceil - 1.
\end{equation}
\end{lemma}

\begin{@proof}
Appendix A. 
\end{@proof}

\begin{lemma}\label{lemma:round-termination}
Round $i$ terminates within $N \leq N_i^{max}$ shots, where 
\begin{equation}
    N_i^{max} := \frac{2}{\sin^2\left(\frac{\pi}{21}\right)\sin^2 \left(\frac{8\pi}{21}\right) }\ln \left(\frac{2}{\alpha_i}\right).
\end{equation}

\end{lemma}

\begin{@proof}
FindNextK returns the largest odd integer in the range $[3 K_i, (\pi/2)/(\theta_u - \theta_l)]$
such that
\[\left\lfloor \frac{K \theta_l}{\pi/2} \right\rfloor = \left\lceil \frac{K \theta_u}{\pi/2}\right\rceil - 1.\]
If no such $K$ exists, then FindNextK returns the old value of $K$ (which is $K_i$).
We will show that if $N = N_{max}$, then such a $K$ does in fact exist. This means a new $K$ is found and the algorithm will advance to the next round within the allotted number of shots. 

Plugging in $N_{max}$ for $N$ into the expression in Line 17 in Algorithm 1 for 
$\epsilon_{a_i}$ gives that
$\epsilon_{a_i} = \frac 1 2 \sin (\pi/21) \sin(8 \pi/21)$
and therefore $a_i^{max} - a_i^{min} \le \sin (\pi/21) \sin(8 \pi/21)$.
Furthermore, since $\epsilon_{a_i} > 0$, the values of 
$a_i^{max}$ and $a_i^{min}$ are chosen so that $1 \ge a_i^{max} > a_i^{min} \ge 0$,
which in turn implies that
$\gamma_i^{min}$ and $\gamma_i^{max}$ are chosen in Table \ref{table:theta-i-ci}
so that $\pi/2 \ge \gamma_i^{max} > \gamma_i^{min} \ge 0$
and 
\begin{align*}
|\sin^2 (K_i \theta_u) - \sin^2 (K_i \theta_l) | &= |\sin^2 \gamma_i^{max} - \sin^2 \gamma_i^{min}|\\  
&= a_i^{max} - a_i^{min}.
\end{align*}

Also, since $\pi/2 \ge \gamma_i^{max} > \gamma_i^{min} \ge 0$, the values of
$\theta_u$ and $\theta_l$ are chosen in Lines 21 and 22 so that
\[\left\lfloor \frac{K_i \theta_l}{\pi/2} \right\rfloor = \left\lceil \frac{K_i \theta_u}{\pi/2}\right\rceil - 1.\]
We can now apply Lemma \ref{lemma:increase-angle} with
$\theta_b = K_i \theta_u$ and $\theta_a = K_i \theta_l$.
Therefore, there is an odd integer in the range
from $3$ to $(\pi/2)/(\theta_u - \theta_l)$ such that
\[\left\lfloor \frac{q K_i \theta_l}{\pi/2} \right\rfloor = \left\lceil \frac{q K_i \theta_u}{\pi/2}\right\rceil - 1.\]
Since $q$ and $K_i$ are both odd integers, $q K_i$ is also an odd integer.
Since $q \ge 3$, then $q K_i \ge 3 K_i$.
Since $q \le (\pi/2)/(\theta_b - \theta_a)$,
then $q K_i \le (\pi/2)/(\theta_u - \theta_l)$.
Therefore the integer $q K_i$ is a viable candidate for a new $K$ which will allow the algorithm
to advance to the next round. 
\end{@proof}

\subsection{Proof for the Probability of Success. }

The sequence $K_0, \ldots, K_t$ has the property that each $K_i$ is in the range from $1$ to $K_{max} = \pi/4 \epsilon$
and $K_i \ge 3 K_{i-1}$. The following lemma will assist in the analysis by allowing us to assume  fixed worst-case values
for the $K_i$'s. The proof of Lemma \ref{lemma:K-i-sum} is given in the appendix.

\begin{lemma}\label{lemma:K-i-sum}
Consider a finite sequence $K_0 = 1, K_1 \ldots, K_t$ such that $K_t \le K_{max}$ and $K_i \ge 3 K_{i-1}$ for all $1 \le i \le t$.
Suppose that $f$ is a function that is positive and increasing over the range $[1, K_{max}]$, then
\begin{equation}\label{eq:increasing-sum}
    \sum_{i=1}^t f(K_i) \leq \sum_{i=0}^{t-1} f\left(\frac{K_{max}}{3^{i}}\right).
\end{equation}
\end{lemma}

\begin{proof}
    Appendix B.
\end{proof}

We can now show that the algorithm succeeds with the desired probability. Lemma \ref{lemma:union-bound} 
establishes that with probability at least $1 - \alpha$,
the true value of $\theta_a$ lies within the final interval
$[\theta_l, \theta_u]$, where $\theta_u - \theta_l \le 2 \epsilon$. This in turn implies that the true value of $a = \sin^2 \theta_a$
lies in the interval $[a_l, a_u] = [\sin^2(\theta_l), \sin^2(\theta_u)]$.
Finally, Lemma \ref{lemma:confint} establishes that $[a_l, a_u]$ is sufficiently small to give the desired bound
on the approximation for $a$, namely that $a_u - a_l \le 2 \epsilon$.

\begin{lemma}\label{lemma:union-bound}
With probability $\geq 1 - \alpha$ the algorithm returns an estimate satisfying 
$\theta_a \in [\theta_l, \theta_u]$.
\end{lemma}

\begin{@proof}
There are a total of $t$ rounds in which the index $i$ goes from $1$ to $t$.
Consider the last iteration of the inner loop in  round $i$, and let $\theta_l^{(i)}$ and
$\theta_u^{(i)}$ be the value of $\theta_l$ and $\theta_u$ as assigned in Lines 21 and 22 of the algorithm.
The round ends either because FindNextK succeeds in finding a new value for $K$ or because $\theta_u - \theta_l  < 2 \epsilon$
in which case the algorithm terminates. The final values $\theta_u$ and $\theta_l$ are equal to
$\theta_l^{(t)}$ and
$\theta_u^{(t)}$.
Define $E_i$ to be the event that $\theta_a \in [\theta_l^{(i)}, \theta_u^{(i)} ]$.
Note that the algorithm starts with $\theta^{(0)}_u = \pi/2$ and $\theta^{(0)}_l = 0$,
so $\mathbb{P}[E_0] = 1$.

We will first show that the probability that $\mathbb{P}[\overline{E_i} \mid E_{i-1}] \le \alpha_i$.
At the beginning of round $i$, the current values for $\theta_l$ and
$\theta_u$ are $\theta_l^{(i-1)}$ and $\theta_u^{(i-1)}$, respectively.
Furthermore, as a result of the conditions of 
FindNextK at the end of round $i-1$, 
$\lceil K_i \theta_u / (\pi/2) \rceil - 1 = \lfloor K_i \theta_l / (\pi/2) \rfloor = R_i$.
If $\theta_a \in [ \theta_l^{(i-1)}, \theta_u^{(i-1)}]$, then 
$K_i \theta_a$ can be expressed as $R_i (\pi/2) + \gamma_a$, where
$0 \le \gamma_a \le \pi/2$. 
The measurements $Q^{k_i}\mathcal{A}\ket{0}_n\ket{0}$ return a $1$ with probability
$a_i = \sin^2(K_i \theta_a)$ which is equal to $\sin^2(\gamma_a)$ or $1 - \sin^2 (\gamma_a)$ depending
on whether the quadrant $R_i$ is odd or even.
After $N$ measurements, we have an estimate for $a_i$, which is denoted by $\hat{a}_i$.
By the Chernoff-Hoeffding bound for i.i.d. Bernoulli trials with $N$ samples,  
\begin{equation}\label{eq:ch-bound}
    \mathbb{P}[a_i \not \in [\hat{a}_i - \epsilon, \hat{a}_i + \epsilon]] \leq 2 \exp (-2N \epsilon^2).
\end{equation}
Selecting $\epsilon = \sqrt{(1/2N) \ln(2/\alpha_i)}$ gives a probability at most $\alpha_i$.
The angles $\gamma_{min}$ and $\gamma_{max}$ are chosen so that the interval between
$\sin^2 (\gamma_{min})$ and $\sin^2 (\gamma_{max})$ is equal to 
$[\hat{a}_i - \epsilon, \hat{a}_i + \epsilon]$. Therefore, the probability that
$\gamma_a$ is not in $[\gamma_{min}, \gamma_{min}]$ is at most $\alpha_i$. 
If $\gamma_a \in [\gamma_{min}, \gamma_{min}]$, then assigning
$$\theta_l \leftarrow \frac{R_i (\pi/2) + \gamma_{min}}{K_i}~~~~~~~~\theta_u \leftarrow \frac{R_i (\pi/2) + \gamma_{max}}{K_i}$$
guarantees that $\theta_a \in [\theta_l, \theta_u]$.

Now we can use the union bound to show that the total error probability is bounded by $\alpha$ as required by Algorithm \ref{alg:modified-iqae}. 
\begin{align}
    \mathbb{P}[\theta_a \not \in [\theta_l, \theta_u]] &= \mathbb{P}[\overline{E_t}]\\  &\leq \mathbb{P}[\exists i \in \{1, \ldots, t\} : \overline{E_i} \wedge E_{i-1}] \\ 
    &\leq \sum_{i=1}^t \mathbb{P} [\overline{E_i} \wedge  E_{i-1}] \\ 
    &\leq \sum_{i=1}^t \mathbb{P} [\overline{E_i} \mid  E_{i-1}]  \leq \sum_{i=1}^t \alpha_i
\end{align}
Now plugging in the chosen values for $\alpha_i$ and using Lemma \ref{lemma:K-i-sum} with $f(x) = x$, we get that 
\begin{align}
    \sum_{i=1}^t \alpha_i ~~&=~~ \frac{2\alpha}{3K_{max}} \sum_{i=1}^t K_i \\ ~~&\le~~ \frac{2\alpha}{3K_{max}} \sum_{i=0}^{t-1} \frac{K_{max}}{3^{i}} \\ 
~~&=~~ \frac{2\alpha}{3} \sum_{i=0}^{t-1} \frac{1}{3^{i}} \\ ~~&\leq~~ \frac{2\alpha}{3} \frac{3}{2} ~~=~~ \alpha
\end{align}

Therefore the probability of failure is bounded above by $\alpha$. 
\end{@proof}

\begin{lemma}\label{lemma:confint}
A confidence interval for $\theta_a$ with width at most $2\epsilon$ is sufficient to get a confidence interval for $a$ with width at most $2\epsilon$. 
\end{lemma}

\begin{@proof}
Given a confidence interval $[a_l, a_u]$ for $a$, we can derive a relation to the confidence interval for the parameter $\theta_a$ using $a = \sin^2(\theta_a)$, and defining $\sin^2(\theta_l) := a_l$ and $\sin^2(\theta_u) := a_u$. 

\begin{align*}
    \frac{|a_u - a_l|}{2} &= \frac{|\sin(\theta_u +\theta_l)||\sin(\theta_u - \theta_l)|}{2} \\ 
    &\leq \frac{|\theta_u - \theta_l|}{2}.
\end{align*}

So to achieve a confidence interval for $a$ that is smaller than $\epsilon$, i.e. $|a_u - a_l|/2 \leq \epsilon$, it suffices to achieve $|\theta_u - \theta_l|/2 \leq \epsilon$ for our estimate of $\theta_a$.
\end{@proof}

\subsection{Bounding the Total Number of Queries. }

We have proven that the algorithm does indeed succeed with the desired probability. In the next lemma we prove an upper bound on the total number of queries required, which concludes the proof of Theorem \ref{thm:algorithm}.

\begin{lemma}
Algorithm 1 requires at most $\frac{62}{\epsilon}\ln\frac{6}{\alpha}$ applications of $A$, achieving the desired query complexity of $O\left(\frac{1}{\epsilon} \log \frac{1}{\alpha}\right)$. 
\end{lemma}

\begin{@proof}
The query complexity is defined as the number of times the operator $Q$ is applied. Each circuit $Q^{k_i}\mathcal{A}\left|0\right>_n\left|0\right>$ applies the operator $Q$ a total of $k_i$ times, and in each round we measure the state created by this state at most $N_i^{max}$ times. For simplicity, we define the constant $C = 1/\left(\sin^2\left(\pi/21 \right)\sin^2\left(8\pi/21 \right)\right)$.

\begin{align}
    N_{queries} &\leq \sum_{i=1}^t k_iN_i^{max} \\ 
    &< \frac{1}{2} \sum_{i=1}^t K_iN_i^{max} \\ 
    &= C\sum_{i=1}^t K_i\ln\left(\frac{2}{\alpha_i}\right) \\ 
    &= C\sum_{i=1}^t K_i\ln\left(\frac{3K_{max}}{K_{i}\alpha}\right) \label{eq:sum-to-upperbound}
\end{align} 
To invoke Lemma \ref{lemma:K-i-sum}, notice that $f(x) = x\ln\left(\frac{c}{x}\right)$ is increasing for all $x \leq c/e$: 
\begin{align}
    x \leq \frac{c}{e} \Rightarrow 1 \leq \ln\left(\frac{c}{x}\right) \Rightarrow f'(x) &= \ln\left(\frac{c}{x}\right) - 1 \geq 0.
\end{align}
To apply the lemma to the expression in Equation (\ref{eq:sum-to-upperbound}), we set $c = 3 K_{max}/\alpha$ and note that for $1 \le K_i \le K_{max}$, $K_i \le 3 K_{max}/\alpha e$.
So we can use Lemma \ref{lemma:K-i-sum} to upperbound the sum:

\begin{figure*}[!ht]
    \centering
    \includegraphics[width=.85\linewidth]{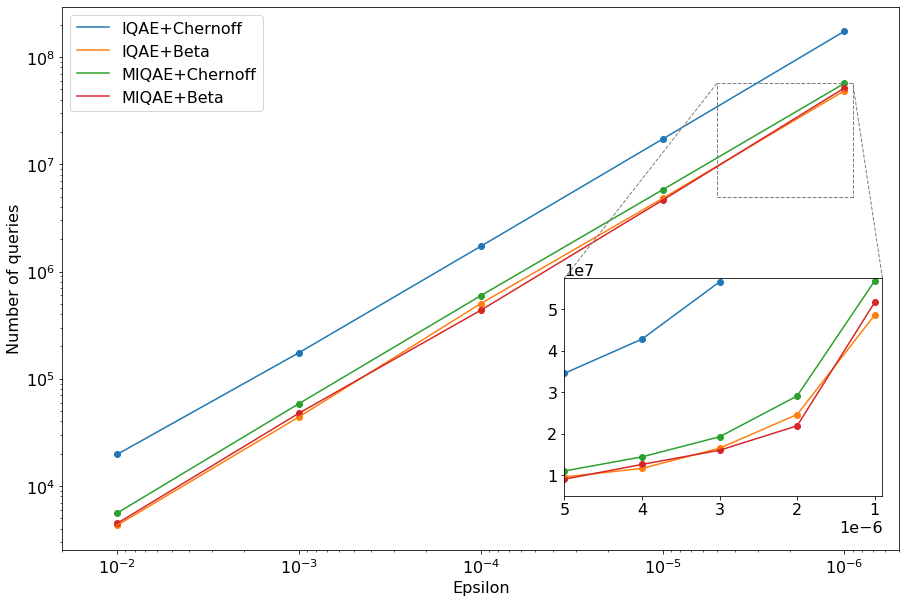}\\
    \caption{\em Plots comparing the target precision $\epsilon$ with number of queries made to $\mathcal{A}$. The main graph is drawn in a log-log plot over the full range of $\epsilon$'s tested, and the close-up uses a log-lin plot to increase the visibility of the respective lines. As described in the main text, the amplitudes are averaged over varied amplitudes. The other two hyperparameter selections are $\alpha = 0.05$, and $N_\text{shots} = 100$.}
    \label{fig:results}
\end{figure*}

\begin{figure*}[!ht]
    \centering
    \includegraphics[width=.85\linewidth]{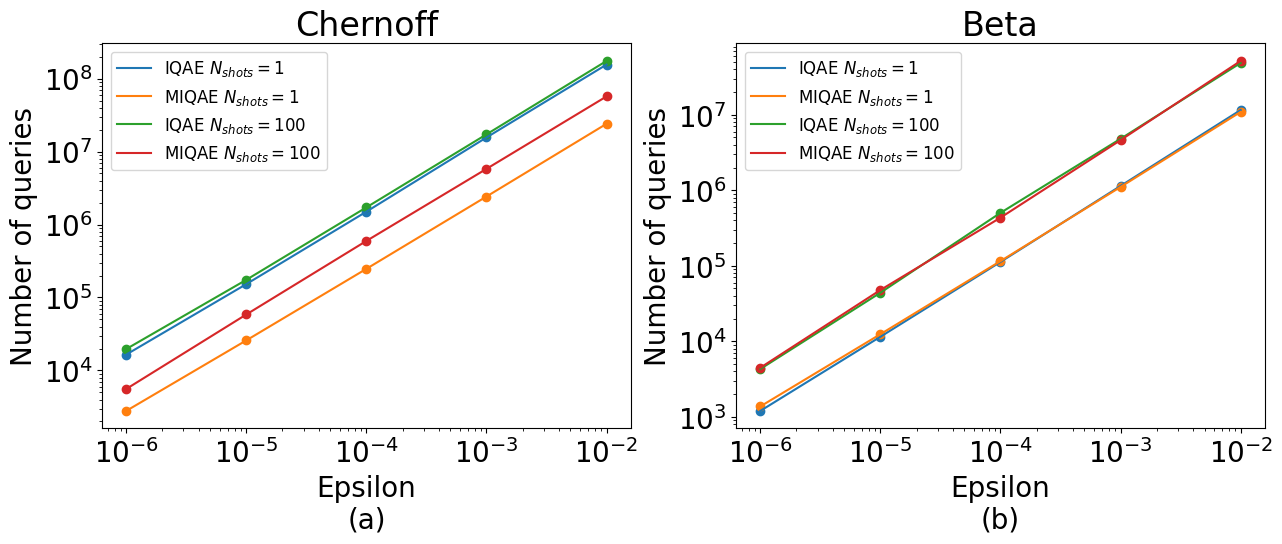}\\
    \caption{\em Experimental results comparing the numerical query complexity when varying the hyperparameter $N_{shots}$. Regardless of the confidence interval computation method of choice, we see that there is a large improvement in numerical performance across the IQAE variants tested. In (a), the comparison is made between the two variants using the Chernoff-Hoeffding bound, whereas in (b) we used the Clopper-Pearson method (Beta). A greater improvement is seen with the Chernoff method for MIQAE, while the improvement is comparable when using the Beta method.}
    \label{fig:results-nshots-1}
\end{figure*}

\begin{align}
 & C\sum_{i=1}^t K_i\ln\left(\frac{3K_{max}}{K_{i}\alpha}\right) \\
 & \leq C \sum_{j=1}^{t-1} \frac{K_{max}}{3^j} \ln\left(\frac{3^{j+1}K_{max}}{K_{max}\alpha}\right) \\ 
    &= C \sum_{j=0}^{t-1} \frac{K_{max}}{3^j} \ln\left(\frac{3^{j+1}}{\alpha}\right) \\
    &= CK_{max}\left( \ln\left(\frac{3}{\alpha}\right) \sum_{j=0}^{t-1}  \frac{1}{3^j} + \sum_{j=0}^{t-1} \frac{1}{3^j} \ln\left(3^{j}\right)  \right) \\
    &= CK_{max}\left( \ln\left(\frac{3}{\alpha}\right) \sum_{j=0}^{t-1}  \frac{1}{3^j} + \ln\left(3 \right)\sum_{j=0}^{t-1} j \frac{1}{3^j}   \right) \\
    &\leq CK_{max}\left(\ln\left(\frac{3}{\alpha}\right) \frac{3}{2} + \frac{3 \ln 3}{4} \right) \\
    &= \frac{3}{2} CK_{max} \ln \left( \frac{\sqrt{27}}{\alpha} \right) \\ 
    &= \frac{3\pi C}{8} \frac{1}{\epsilon}\ln \left( \frac{\sqrt{27}}{\alpha} \right) \label{eq:query-complexity-exact}\\
    &\approx \frac{62}{\epsilon} \ln \left( \frac{6}{\alpha} \right) \\
    &= O\left(\frac{1}{\epsilon}\ln \left( \frac{1}{\alpha}\right)\right)
\end{align}

\end{@proof}

\subsection{Absolute and Relative Error. }

Algorithm \ref{alg:modified-iqae} can be modified to output an approximation of $a$ to within a $(1 + \epsilon)$ relative error as well, by adding an extra $1/a$ overhead. To accomplish this, an extra outer loop is added to Algorithm \ref{alg:modified-iqae} such that the algorithm will terminate when $a > 0.5$. If the algorithm completes the inner loop without achieving our desired approximation, the upper bound is continually halved until the algorithm successfully approximates the targeted amplitude to the correct precision. Details of the modification and proof are outlined in Appendix C.

%% file: experiments.tex
\section{Results}

\definecolor{mplblue}{HTML}{1f77b4}
\definecolor{mplorange}{HTML}{ff7f0e}
\definecolor{mplgreen}{HTML}{2ca02c}
\definecolor{mplred}{HTML}{d62728}

\begin{figure*}[!ht]
\centering
     \begin{subfigure}{0.3\textwidth}
         \centering
         \includegraphics[width=.95\textwidth]{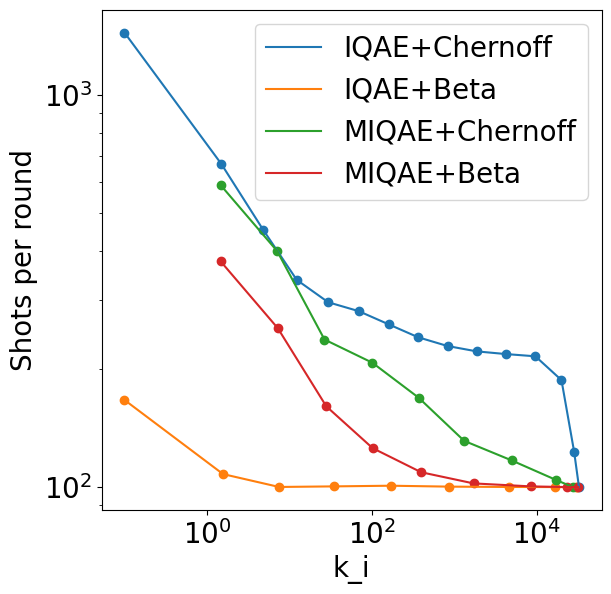}
         \caption{\em Number of measurements used as a function of $k_i$.}
     \end{subfigure}
     \hfill
     \begin{subfigure}{0.3\textwidth}
         \centering
         \includegraphics[width=.95\textwidth]{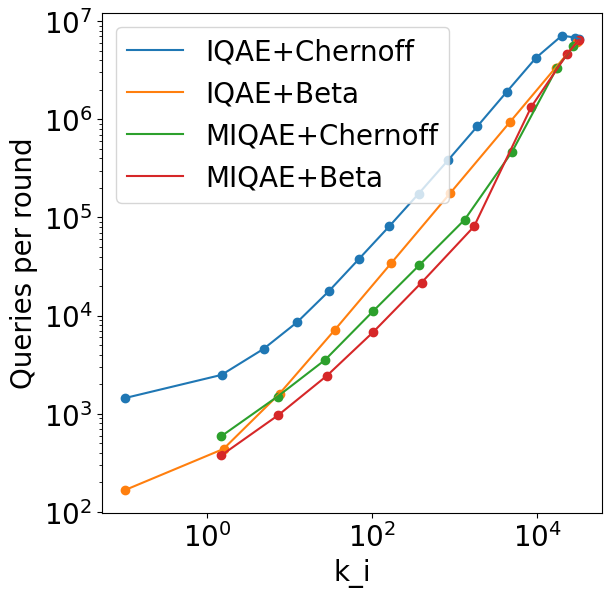}
         \caption{\em Number of queries used as a function of $k_i$.}
         \label{fig:queries_per_ki}
     \end{subfigure}
     \hfill
     \begin{subfigure}{0.3\textwidth}
         \centering
         \includegraphics[width=.95\textwidth]{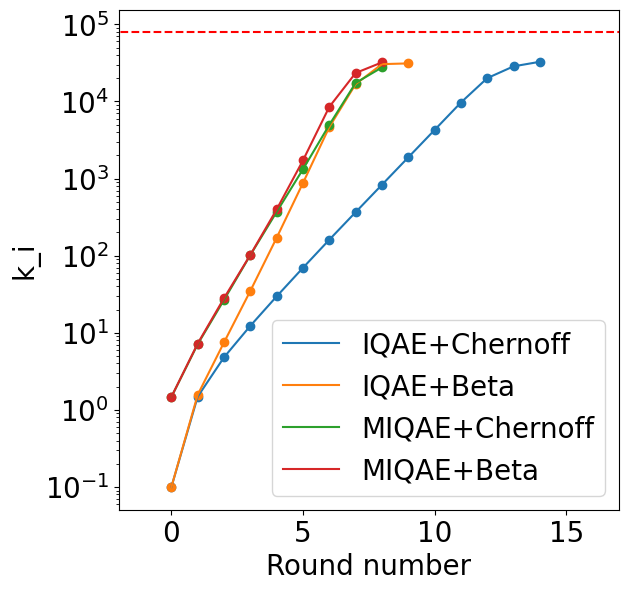}
         \caption{\em Scaling of $k_i$ as a function of the round number.}
         \label{fig:ki_per_round_num}
     \end{subfigure}
    \caption{\em Per-round results from running the algorithms with $\epsilon = 10^{-5}$, averaged over 10,000 runs. The red dotted line in the rightmost graph represents $k_{\text{max}}$.}
    \label{fig:results-per-round}
\end{figure*}

In this section, we present the results of Algorithm \ref{alg:modified-iqae}. We are particularly interested in visualizing the empirical query complexity, and comparing our performance to the original IQAE algorithm. Our code can be found at https://github.com/shifubear/Modified-IQAE. 

Our experiments were run over the 3-tuples of parameters from the Cartesian product of the sets $\{\texttt{epsilons},~\texttt{amplitudes},~\texttt{confint\_methods}\}$. Here, $\texttt{epsilons} = \{10^{-2}, 10^{-3}, \cdots, 10^{-6}\}$, $\texttt{amplitudes} = \{0, \frac{1}{16}, \cdots, 1\}$, and $\texttt{confint\_methods} = \{\texttt{chernoff}, \texttt{beta}\}$. The term {\tt beta}
refers to  the Clopper-Pearson method for determining the confidence interval. The confidence level $\alpha$ is set to $.05$ for all experiments.

Figure \ref{fig:results} shows the query complexity of the four IQAE algorithms: Modified IQAE (MIQAE) and the original IQAE with the two difference methods for determining the confidence interval (Chernoff and Beta).
Experiments were run using the parameter space described above.
The plot shows the total number of queries as a function of the precision parameter  epsilon. To discount any anomalies that arise from random sampling, we repeat each parameter combination 10,000 times and average the total number of queries per $\epsilon$. We also average the results over all amplitudes, since their query complexity should not change based on the ground-truth amplitude. While testing the original IQAE algorithm, we found that amplitudes expressed using powers of 2 (i.e. $0$, $\frac12$ ...) were noticeably easier to estimate. We slightly perturb the ground truth amplitude by sampling $\Tilde{a} \sim \mathcal{N}(a, \frac{1}{320})$. Here, the standard deviation is chosen such that $\Tilde{a}$ does not vary more than 5\% of our 1/16 step distance from $a$. This reduces any biases in performance due to ground-truth amplitudes that use powers of 2.

Using the Chernoff confidence interval, there is a distinct improvement in both the query complexity and the total number of queries used. The modified IQAE algorithm (\textcolor{mplgreen}{green}) always makes fewer queries by about half an order of magnitude, when compared to the original (\textcolor{mplblue}{blue}). 

On the log-log scale, the plotted query complexity for \texttt{MIQAE+Beta} and \texttt{IQAE+Beta} are nearly identical. For finer detail, a second set of error tolerances $\{10^{-4}, 7.5\times10^{-5}, \cdots , 10^{-6}\}$ is also tested with the same amplitudes. On a linear-log scale, the original algorithm using the Clopper-Pearson (Beta) interval (\textcolor{mplorange}{orange}) appears to use marginally less queries for higher values of $\epsilon$.

It is possible that the optimal value for each intermediate confidence level $\alpha_i$ using Clopper-Pearson intervals takes a different form, since $\alpha_i$ is specifically designed to optimize the Chernoff confidence interval. However, an analytic proof using the beta distribution is intractable \cite{grinko2021a} and we leave these potential improvements to be discovered.

In Figure \ref{fig:results-nshots-1}, we calculate the numerical query complexity when varying the hyperparameter $N_{shots}$. As described in an earlier section, this parameter determines how many samples we take from the quantum computer in a given iteration before using the results to recompute a confidence interval. As demonstrated in the plot, there is a substantial improvement for both IQAE and MIQAE regardless of the method of computing the confidence interval, when $N_{shots} = 1$. This is true even when taking into account IQAE's no-overshooting condition, showing that this sampling method is the most general way to extend that heuristic. Because of this, unless there are some resource constraints around how the quantum computer and classical computer communicate that would benefit from batch sampling, we suggest setting this parameter to 1 for the lowest query complexity. 

Furthermore, there is a substantial improvement to the query complexity when using MIQAE with the Chernoff-Hoeffding bound as demonstrated in Figure \ref{fig:results-nshots-1} (a). Since the analysis in the paper was centered around this case, this may suggest that our algorithm doesn't necessitate the no-overshooting condition proposed in IQAE \cite{grinko2021a} and ChebAE \cite{rall2022}. Instead, the gradual relaxation for the confidence interval requirement implicitly serves this purpose, leading to an improved numerical performance. We comment that this may mean a careful analysis of the algorithm using the Clopper-Pearson (Beta) method could lead to a similar improvement using that technique. 

We also analyze the behavior of the algorithms in each individual round. To accomplish this, we record the number of shots $N$, the Grover power $k_i$, and the number of queries made at each $k_i$. As shown in Figure \ref{fig:results-per-round}, the number of queries per $k$ increases at the same rate for all combinations of algorithms and confidence intervals. Therefore, the modified algorithm must reduce the total queries made by using fewer shots at higher values of $k_i$, as Aaronson and Rall observed $\cite{aaronson2020}$. This behavior can also be experimentally verified in Figure \ref{fig:results-per-round}, where the modified IQAE uses progressively fewer shots for higher values of $k_i$. 

\begin{figure*}[!ht]
    
    \begin{subfigure}{.48\textwidth}
        \centering
        \includegraphics[width=\linewidth]{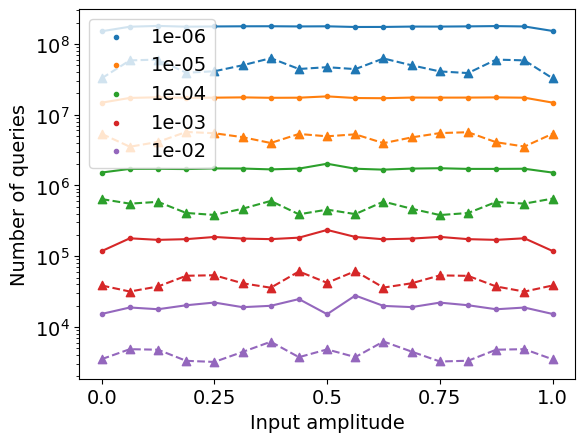}
        \label{fig:results-query-a_a}
        \caption{Original IQAE}
    \end{subfigure}
    \begin{subfigure}{.48\textwidth}
        \centering
        \includegraphics[width=\linewidth]{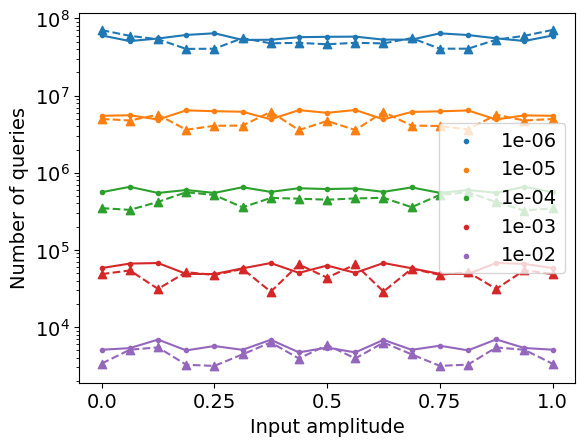}
        \label{fig:results-query-a_b}
        \caption{Modified IQAE}
    \end{subfigure}
    \caption{\em Query count vs. amplitude for each combination of algorithm and confidence interval method, recorded at different values of $\epsilon$ described in the legends. The solid lines and dotted lines correspond to the Chernoff-Hoeffding and Clopper-Pearson confidence intervals respectively. On average, the number of queries do not depend on the input amplitude in either algorithm.}
    \label{fig:results-query-a}
\end{figure*}

\begin{figure}[!ht]
    \centering
    \includegraphics[width=\linewidth]{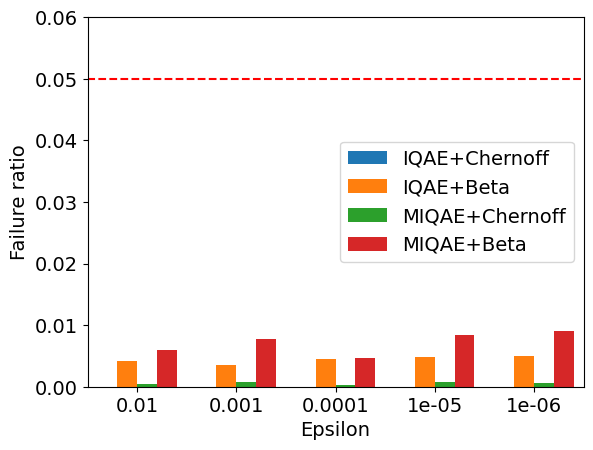}
    \caption{\em The frequency of incorrect estimations (true $a$ not within the estimated confidence interval) using the parameter space described in the beginning of this section. All frequencies are well below the dotted line representing $\alpha = .05$.} 
    \label{fig:results-failures}
\end{figure}

Figure \ref{fig:results-per-round}(c) illustrates the values of the Grover power $k$ in each round. In all runs, the value of $k$ never surpasses $k_{\text{max}}$, supporting the correctness of the IQAE algorithm. $k$ increases most slowly between rounds with \texttt{IQAE+Chernoff}, whereas all other configurations increase their $k$'s at similarly fast rates. This explains the results that are shown in each configuration's query complexity (Figure \ref{fig:results}), where the \texttt{IQAE+Chernoff} uses the most queries and all other configurations use a similar number of queries. Evidently, the evolution of $k$ is central to the algorithm's overall performance. However, it remains unclear whether this relationship is responsible for the improvement seen in the other algorithms over \texttt{IQAE+Chernoff}, or a side effect of controlling the evolution of $N$.

Figure \ref{fig:results-query-a} plots the number of queries made for each 3-tuple of experiment parameters against the ground-truth input amplitude. Like Figure \ref{fig:results}, the results are averaged across 10,000 runs. The number of queries for any amplitude generally stays the same in both the modified and original algorithms for a specific $\epsilon$. The query count for algorithms using Clopper-Pearson intervals appears to have higher variance between amplitudes, but still averages out to the same level trend. The performance of the Chernoff interval also closely matches that of the Clopper-Pearson interval in our modified algorithm. When using the original IQAE algorithm, however, there is a significant gap between the algorithm's query count depending on which confidence interval method is used. These observations justify our claim that input amplitudes do not affect the algorithm's query complexity.

Finally, we verified that the algorithms return correct estimations with the desired frequency. According to the proofs, both IQAE algorithms should return an interval that contains the ground-truth amplitude with probability at least $1-\alpha$. In other words, both algorithms should return an interval that does not contain the true amplitude
with frequency at most $\alpha$. We use our previous data, which tests all error tolerance, amplitude, and confidence interval combinations using $\alpha = .05$. We plot the number of failures for a specific amplitude and confidence interval method for each input $\epsilon$ in Figure \ref{fig:results-failures}. Overall, the modified algorithm has a higher failure rate than the original, but all failure rates are clearly well below the theoretical bounds.

%% file: conclusion.tex
\section{Conclusion and Future Work}

In this paper, we show that a modification to Grinko et al.'s IQAE algorithm produces an asymptotically optimal QFT-free quantum amplitude estimation algorithm, and furthermore demonstrate experimentally that this modification is not only a theoretical improvement, but makes the algorithm more suitable to run in the upcoming NISQ era. 
Future work may include experimenting with this algorithm on real quantum hardware, and analyzing its performance in the presence of noise and noise mitigation techniques.

%% file: appendix.tex
\newpage
\section{Appendix A: Proof of Lemma \ref{lemma:increase-angle}}

\begin{figure}[!b]
    \centering
    \includegraphics{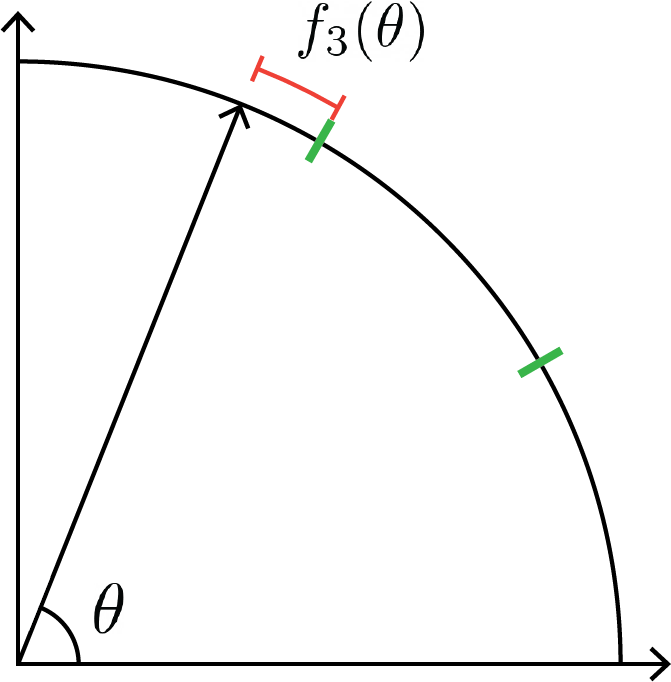}
    \caption{\em An illustration describing the quantity computed by $f_3$. We only show the first quadrant for illustration purposes. The circular arc is subdivided into 3 sections (green markers). $f_3(\theta)$ returns the angle corresponding to the red arc, describing the angular distance to the nearest subdivided marker from $\theta$.}
    \label{fig:f3}
\end{figure}

This proof follows very closely from the proof of Lemma 1 in \cite{grinko2021a}. 

\begin{@proof}
This lemma characterizes the conditions in which a pair of angles that lie in the same quadrant can be multiplied by an odd integer $q$ such that the resulting angle still lies in a single quadrant. 

We begin by describing a geometric intuition on when such an odd integer can be found. Consider a circle in which each quadrant has been subdivided into $j \in [3, 5, 7]$ slices. If the current interval $[\theta_a, \theta_b]$ lies within one of these slices, then multiplying that interval by $j$ will guarantee that the scaled interval will still lie inside of a quadrant. 

We now define a set of piecewise-linear functions $f_j$ on $[0, \pi/2]$ for integers $j \in [3, 5, 7]$ to capture this geometric intuition as follows. There are a total of $j$ pieces, and the $m$-th piece where $m \in \{1, \ldots, j\}$ is constructed on the interval $\left[\frac{(m-1)\pi}{j}, \frac{m\pi}{j}\right]$. On this interval, the function of interest is defined as 

\begin{figure*}[!ht]
    \centering
    \includegraphics[width=\textwidth]{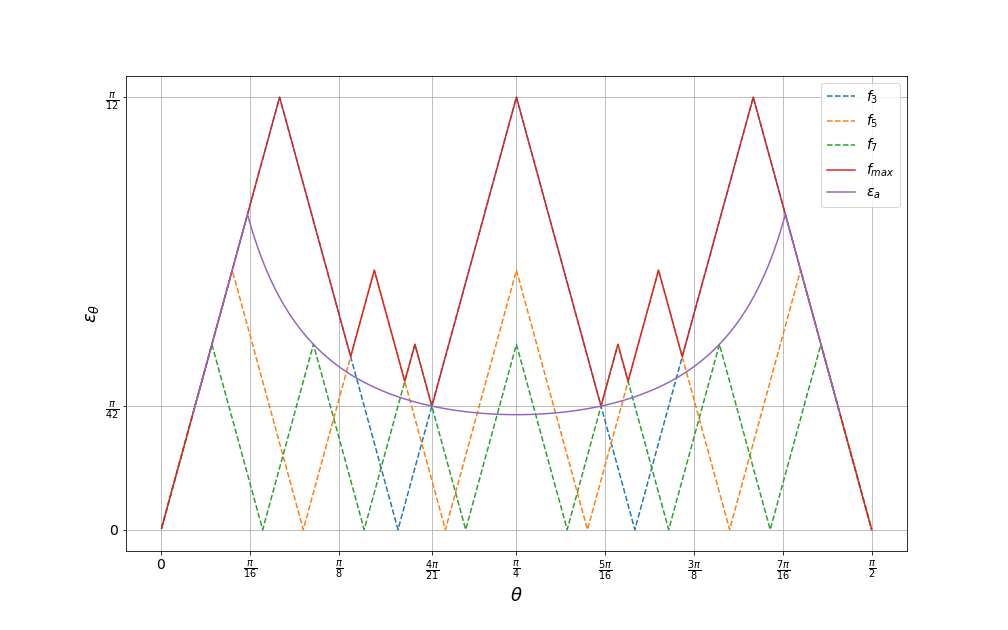}
    \caption{\em A plot of the functions $f_3$, $f_5$, $f_7$, $f_{max}$, and $\epsilon_\theta$. By construction, $\epsilon_\theta$ is a valid lower bound to $f_{max}$, with an intersection at $\theta = \frac{4\pi}{21}$. }
    \label{fig:fmaxeps}
\end{figure*}

\begin{equation}
    f_j(\theta) := \min\left\{ \theta - \frac{(m-1)\pi}{2j}, \frac{m\pi}{2j} - \theta\right\}.
\end{equation}

Figure \ref{fig:f3} illustrates the quantity being computed in the case that $j = 3$, and similar images can be used to visualize the case for other values of $j$ as well. This set of functions characterizes the closest distance from a point on a circle to one of the $j$ subdivided points in the quadrant it lies in. By combining these functions, it is possible to find the maximum possible interval width centered at the input point such that the scaled interval still lies within a quadrant. The combined function is as follows:   

\begin{equation}
    f_{max}(\theta) := \max\{f_3(\theta), f_5(\theta), f_7(\theta)\}.
\end{equation}

What remains is to show that $|\sin^2 \theta_b - \sin^2 \theta_a| \leq \sin \frac{\pi}{21} \sin \frac{8\pi}{21}$ is sufficient to guarantee that $\frac{|\theta_b - \theta_a|}{2} \leq f_{max}(\theta)$ where $\theta := (\theta_b + \theta_a)/2$ is the midpoint of the interval. We first introduce some notation 

\begin{align}
    \epsilon_a &:= \frac{|\sin^2 \theta_b - \sin^2 \theta_a|}{2} \label{eq:epsilon-a} \\ 
    \epsilon_\theta &:= \frac{|\theta_b - \theta_a|}{2}. 
\end{align}

The following derivation relates $\epsilon_a$ with $\theta$ and $\epsilon_\theta$:

\begin{align}
    \epsilon_a &:= \frac{|\sin^2 \theta_b - \sin^2 \theta_a|}{2}\\ 
    &= \frac{|\sin(\theta_b + \theta_a)||\sin(\theta_b - \theta_a)|}{2}\\
    &= \frac{|\sin(2\theta)||\sin(2\epsilon_\theta)|}{2}. \label{eq:epsilon-theta-to-epsilon-a}
\end{align}

Without loss of generality, the remaining analysis will assume that $\theta_a, \theta_b$ lie in the first quadrant. The same analysis holds for the other quadrants by symmetry. With this assumption, our analysis of (\ref{eq:epsilon-theta-to-epsilon-a}) will proceed without its absolute values. From $|\sin^2 \theta_b - \sin^2 \theta_a| \leq \sin \frac{\pi}{21} \sin \frac{8\pi}{21}$ and (\ref{eq:epsilon-a}), we have an upperbound for $\epsilon_a$ given by $\epsilon_a \leq \frac{1}{2}\sin\frac{\pi}{21}\sin\frac{8\pi}{21}$. Given this upperbound, the goal is to determine the maximum tolerable $\epsilon_\theta$ that still allows us to scale the interval such that it lies within a quadrant. To achieve this, we define the following function which is constructed by inverting (\ref{eq:epsilon-theta-to-epsilon-a}). 

\begin{equation}
    \epsilon_{\theta}(\epsilon_{a}, \theta) := \min\left(  \frac{1}{2} \arcsin \left( \frac{2\epsilon_{a}}{\sin(2\theta)} \right), \theta, \pi/2 - \theta \right). \label{eq:epsilon-theta-fn-1}
\end{equation}

The right two terms in the minimum function follow from the constraint that the interval must lie within a single quadrant. Naturally, this constraint forces the interval to be upperbounded by the boundaries of the quadrant. Since we have an upperbound on $\epsilon_a$, the following is the maximum value that $\epsilon_\theta$ can take:

\begin{equation}
    \epsilon_{\theta}(\theta) := \min\left(  \frac{1}{2} \arcsin \left( \frac{\sin \frac{\pi}{21} \sin \frac{8\pi}{21}}{\sin(2\theta)} \right), \theta, \pi/2 - \theta \right). \label{eq:epsilon-theta-fn-2}
\end{equation}

We want to show that $\epsilon_\theta \leq f_{max}$ for all values of $\theta$ in $[0, \pi/2]$, which means that whenever the size of the interval $[\theta_a, \theta_b]$ is $\epsilon_\theta$, there exists an odd integer $q \geq 3$ such that the interval $[q\theta_a, q\theta_b]$ lies within a single quadrant. The functions $\epsilon_\theta(\theta)$ and $f_{max}(\theta)$ evaluated near the endpoints of the interval $[0, \pi/2]$ are identical linear functions. 

What remains is to show that the inequality holds for the middle section of the function, defined by the first term inside the minimum function from (\ref{eq:epsilon-theta-fn-2}). The interval where this term takes the minimum can be calculated by solving for the values of $\theta$ where the functions are equal. This interval is 

\begin{equation}
    \theta \in \left[ \frac{1}{2}\arcsin \sqrt{S}, \frac{\pi}{2} - \frac{1}{2}\arcsin \sqrt{S}  \right],
\end{equation}

where $S := \sin\frac{\pi}{21}\sin\frac{8\pi}{21}$.

A simple analysis of the piecewise linear function $f_{max}$ shows that it takes its minimum value of $\frac{\pi}{42}$ when $\theta = \frac{4\pi}{21}$. Evaluating $\epsilon_\theta$ at this value of $\theta$ gives

\begin{align}
    \epsilon_\theta\left(\frac{4\pi}{21}\right) &= \frac{1}{2}\arcsin\left(\frac{\sin \frac{\pi}{21} \sin \frac{8\pi}{21}}{\sin \frac{8\pi}{21}}\right) \\ 
    &= \frac{\pi}{42}.
\end{align}

For all other values of $\theta$, it is clear that $\epsilon_\theta \leq f_{max}$ which can be seen figuratively in Figure \ref{fig:fmaxeps}. This can be verified by testing that $\epsilon_\theta \leq f_{max}$ holds within each of the intervals that define the piecewise linear parts of the function manually. Therefore, $\epsilon_\theta \leq f_{max}$ in the possible values of $\theta$, showing that there indeed exists an odd integer $q \geq 3$ such that the multiplied interval lies in a single quadrant. 

\end{@proof}

\section{Appendix B: Proof of Lemma \ref{lemma:K-i-sum}}

\begin{@proof}
To prove this lemma, we first show that for all $j = 1, \ldots, t$, we have 

\begin{equation}\label{eq:lemma34-ind-goal}
    K_{j} \leq \frac{K_{max}}{3^{t-j}},
\end{equation}

which by the assumption that $f$ is an increasing function implies that 

\begin{equation}
    f(K_{j}) \leq f\left(\frac{K_{max}}{3^{t-j}}\right).
\end{equation}

We can then sum over all $j$ and rewrite the indices of the right-hand side to get 

\begin{equation}
    \sum_{j=1}^{t} f(K_{j}) = \leq \sum_{j=1}^{t} f\left(\frac{K_{max}}{3^{t-j}}\right) = \sum_{i=0}^{t-1} f\left(\frac{K_{max}}{3^j}\right). 
\end{equation}

Now we show how to prove equation \ref{eq:lemma34-ind-goal}. We prove this statement inductively backward over the indices $j$. By definition, we know that $K_{j-1} \leq K_j/3$.
The base case when $j=t$ holds by definition since $K_t \leq K_{max}$.
By inductive hypothesis, assume that $K_{k} \leq \frac{K_{max}}{3^{t-k}}$ 
holds for some $1 < k \leq t$. Then, 

\begin{align}
    K_{k-1} &\leq \frac{K_k}{3} \\
    &\leq \frac{K_{max}}{3^{t-(k-1)}}
\end{align}

which is what we wanted. 
\end{@proof}

\section{Appendix C: Relative Error Algorithm}

A key difference between QACS \cite{aaronson2020} and IQAE \cite{grinko2021a} were the type of error bounds they used to define their algorithms. IQAE used an absolute error, so our modification to the algorithm also adopts the absolute error bound. QACS used a relative error bound, which is a commonly used error bound when defining approximation algorithms. We show here that the relative error bound can be achieved with an extra $1/a$ overhead, which is optimal as shown in QACS. 

\begin{theorem}\label{thm:algorithm-relative}
Given a confidence level $1 - \alpha \in (0, 1)$, a target accuracy $\epsilon > 0$, and an $(n + 1)$-qubit unitary $\mathcal{A}$ satisfying 

\[\mathcal{A}\ket{0}_n\ket{0} = \sqrt{a}\ket{\psi_0}\ket{0} + \sqrt{1-a}\ket{\psi_1}\ket{1}\]

where $\ket{\psi_0}$ and $\ket{\psi_1}$ are arbitrary $n$-qubit states and $a \in (0, 1)$, Algorithm \ref{alg:modified-iqae-rel} outputs a confidence interval for $a$ that satisfies

\[\mathbb{P}[a \not \in [a_l, a_u]] \leq \alpha\]

where $a_u := a(1 +\epsilon)$ and $a_l := a(1 - \epsilon)$, leading to a $(1 + \epsilon)$-approximation for $a$ with a confidence of $1 - \alpha$, using $O\left(\frac{1}{a} \frac{1}{\epsilon}\log \frac{1}{\alpha}\right)$ applications of $\mathcal{A}$. 
\end{theorem}

\begin{algorithm}
\caption{Relative IQAE}
\label{alg:modified-iqae-rel}
\begin{algorithmic}[1]
\Require $\epsilon > 0$, $\alpha > 0$, Unitary $\mathcal{A}$
\State $[a_l, a_u] = [0, 1]$
\State $\epsilon_i = \epsilon$
\State $\hat{a} = 0.5$
\While{$a_u - a_l > 2\hat{a}\epsilon$} 
\State $\epsilon_i = \epsilon_i / 2$
\State $[a_l, a_u] = \text{ModifiedIQAE} ( \epsilon_i, \alpha, \mathcal{A})$
\State $\hat{a} = (a_u - a_l)/2$
\EndWhile
\State \Return $[a_l, a_u]$
\end{algorithmic}
\end{algorithm}

\begin{@proof}  
We know that Algorithm \ref{alg:modified-iqae} given an error parameter $\epsilon > 0$ outputs an absolute error confidence interval $[\hat{a}+\epsilon, \hat{a}-\epsilon]$ that contains $a$ with confidence $1 - \alpha$. In the $i$-th iteration of the loop of algorithm \ref{alg:modified-iqae-rel}, ModifiedIQAE is called using $\epsilon_i := \epsilon/2^{i}$, which will output an interval satisfying
\begin{equation}
    [\hat{a} + \epsilon/2^{i}, \hat{a} - \epsilon/2^{i}].
\end{equation}
For $a \in [1/2^{i}, 1/2^{i-1}]$, this is a $(1 + \epsilon)$ relative error confidence interval. 

From (\ref{eq:query-complexity-exact}), we can write the exact query complexity of ModifiedIQAE given $\epsilon_i$ as a parameter. We define the constant $C = 1/\left(\sin^2\left(\pi/21 \right)\sin^2\left(8\pi/21 \right)\right)$ to simplify the presentation. 
\begin{align}
    &\frac{3\pi C}{4} \frac{1}{\epsilon_i}\ln \left( \frac{\sqrt{27}}{\alpha} \right) \\
    = &\frac{3\pi C}{4} \frac{2^i}{\epsilon}\ln \left( \frac{\sqrt{27}}{\alpha} \right).
\end{align}
If we required $R$ total iterations to achieve the target approximation error, then $a < 1/2^{R - 1}$, and we can sum over this to get the query complexity 
\begin{align}
    &\frac{3\pi C}{4} \frac{1}{\epsilon}\ln \left( \frac{\sqrt{27}}{\alpha} \right) \sum_{i=1}^R 2^i \\
    &= \frac{3\pi C}{4} \frac{1}{\epsilon}\ln \left( \frac{\sqrt{27}}{\alpha} \right) \left(2^{R+1} - 2\right) \\ 
    &< \frac{3\pi C}{\sin^2\left(\frac{\pi}{21}\right) \sin^2\left(\frac{8\pi}{21}\right)}
    \frac{1}{a} \frac{1}{\epsilon}\ln \left( \frac{\sqrt{27}}{\alpha} \right) \\
    &= O\left(\frac{1}{a}\frac{1}{\epsilon}\ln\left( \frac{1}{\alpha} \right)\right)
\end{align}
The inequality follows from $4/a > 2^{R+1}$. Therefore Algorithm \ref{alg:modified-iqae-rel} returns a relative error confidence interval with confidence $\alpha$ in the desired query complexity. 
\end{@proof}